\documentclass[10pt,conference]{IEEEtran}

\AtBeginDocument{
  \providecommand\BibTeX{{
    \normalfont B\kern-0.5em{\scshape i\kern-0.25em b}\kern-0.8em\TeX}}}

\usepackage{amssymb,amsfonts,amsmath}
\usepackage{graphicx}
\usepackage{xcolor}
\usepackage{balance}
\usepackage{algorithm}
\usepackage{algpseudocode}
\usepackage{multicol}
\usepackage{enumitem}
\makeatletter
\renewcommand*{\ALG@name}{Procedure}
\makeatother
\newtheorem{theorem}{Theorem}[section]
\newtheorem{lemma}[theorem]{Lemma}

\newtheorem{example}{Example}
\newcommand{\sq}{\hbox{\rlap{$\sqcap$}$\sqcup$}}
\newcommand{\qed}{\hspace*{\fill}\sq}
\newenvironment{proof}{\noindent {\bf Proof.}\ }{\qed\par\vskip 4mm\par}

\begin{document}
\title{Fort Formation by an Automaton}
\author{
\IEEEauthorblockN{Kartikey Kant}
\IEEEauthorblockA{\textit{Indian Institute of Technology Guwahati}\\
Guwahati, India \\
kartikeykant@gmail.com}
\and
\IEEEauthorblockN{Debasish Pattanayak}
\IEEEauthorblockA{\textit{Indian Statistical Institute}\\
Kolkata, India \\
drdebmath@gmail.com}
\and
\IEEEauthorblockN{Partha Sarathi Mandal}
\IEEEauthorblockA{\textit{Indian Institute of Technology Guwahati}\\
Guwahati, India \\
psm@iitg.ac.in}
}

\maketitle
\begin{abstract}
    Building structures by low capability robots is a very recent research development~\cite{DBLP:journals/algorithmica/CzyzowiczDP20}. A robot (or a mobile agent) is designed as a deterministic finite automaton. The objective is to make a structure from a given distribution of materials (\textit{bricks}) in an infinite grid $Z\times  Z$. The grid cells may contain a brick (\textit{full cells}) or it may be empty (\textit{empty cells}). The \textit{field}, a sub-graph induced by the full cells, is initially connected. At a given point in time, a robot can carry at most one brick. The robot can move in four directions (north, east, south, and west) and starts from a \textit{full cell}. The \textit{Manhattan distance} between the farthest full cells is the \textit{span} of the field.
    We consider the construction of a \textit{fort}, a structure with the minimum span and maximum covered area. On a square grid, a fort is a hollow rectangle with bricks on the perimeter. We show that the construction of such a fort can be done in $O(z^2)$ time -- with a matching lower bound $\Omega(z^2)$ -- where $z$ is the number of bricks present in the environment.
\end{abstract}


\section{Introduction}
\subsection{The Problem}
Czyzowicz \textit{et al.}~\cite{DBLP:journals/algorithmica/CzyzowiczDP20} have initiated a study on building structures by a robot (or mobile agent) designed as a deterministic finite automaton in an infinite grid with materials.
The motivation behind building such a structure stems from many areas, such as the demarcation of land using a robot in inaccessible areas or building structures in space where the materials are limited. 
In these cases, it is important to utilize the available resources to their complete capacity.
This paper explores how a connected structure can be formed enclosing the maximum possible area with a given number of bricks.
We denote the resulting structure as a \textit{fort}.

\subsection{Environment Model}
We are given an infinite oriented grid $Z \times  Z$, which can be viewed as a grid of cells in a two-dimensional plane, with the cells having either horizontal or vertical common edges. Thus, each cell can have four adjacent cells to its North, East, South, and West. We denote a cell having a brick as a \textit{full cell} and a cell without a brick as an \textit{empty cell}. The sub-graph induced by the full cells is called a \textit{field}. Initially, the field is connected, i.e., there exists a rectilinear path between any two full cells, which consists of full cells. However, after operations by the robot, the field may get disconnected. 
The maximally connected sub-graph of a field is called its \textit{component}. The number of cells in a field is called its \textit{size} and the maximum Manhattan distance\footnote{Manhattan distance between two cells at $(x,y)$ and $(x',y')$ is $\lvert x - x'\rvert + \lvert y - y' \rvert$.}
between two cells in a field is called its \textit{span}.
\subsection{Robot Model}

The robot starts at one of the full cells.
It can move in one of the four directions (East, North, West, or South) from a cell. 
It is oriented towards one of the directions, and when we say the robot turns left, it reorients itself locally (A robot facing north turns left to face west).
It can only \textit{observe} the state of its \textit{current} cell, i.e., the cell contains a brick or not.
It can pick up a brick from a \textit{full} cell and drop the brick in an \textit{empty} cell.
It can move with or without carrying a brick (makes the robots \textit{light} or \textit{heavy}).
It can carry at most one brick at once. The robot can move through full cells while carrying a brick.
It is formalized as a Mealy automaton $\mathcal{R} = (X,Y,\mathsection,\delta, \lambda, S_0, S_f)$, where $X = \{e,f\} \times \{l,h\}$ is the input, $Y = \{N,E,S,W\}\times \{e,f\} \times \{l,h\} $ is the output, $\mathsection$ is the set of states, $\lambda$ is the transition function with $S_0$ as initial state and $S_f$ as the final state.  The alphabet denote the following in order: $N$ : North, $E$: East, $S$: South, $W$: West, $e$: empty cell, $f$: full cell, $l$: light robot (not carrying a brick), $h$: heavy robot (carries a brick).
The state of the robot contains its current orientation, the current cell's status, and if the robot carries a brick.
We follow the same model for the robot introduced by Czyzowicz \textit{et al.}~\cite{DBLP:journals/algorithmica/CzyzowiczDP20}.
The robot has a very limited amount of memory, which maintains the set of states. 
Even with limited memory, it can perform a bounded exploration up to a small distance (say 8) and keep the state of the cells encoded in its states.
In other words, at any point in time, the robot can have information about all the cells at a distance at most $8$ from its current cell. 

\subsection{Target Structure}
The largest connected structure with the maximum enclosed area in a square grid is a hollow rectangle, i.e., the bricks are only on the sides of the rectangle.
The target is to construct this hollow rectangle -- denoted as \textit{fort} -- using the bricks available.
A \textit{perfect fort} is the fort having all its sides made up of an equal number of bricks.
The manhattan distance between the diagonally opposite bricks is the \textit{span} of the fort.
This span will always be an even number for a perfect fort. For a perfect fort, it is easy to deduce the relation $z = 2s'$ where $z$ is the total number of bricks in the fort, and $s'$ is the span of the fort.  The fort, shown in Fig.~\ref{figfort8}, is an example of a perfect fort. A perfect fort contains $4m$ bricks, where each \textit{wall} of the fort contains $m+1$ bricks. Note that the bricks at the corner are shared by two \textit{walls}. Any fort consisting of odd number of bricks is a \textit{rough fort} (ref. Fig.~\ref{figstage1}(b), (d), (f) and (h)). The rough fort is a rectangle where an additional brick may be attached to one of its corners since the number of bricks is odd.

\subsection{Challenge}
Designing algorithms for a low memory robot with minimal capability is challenging, given the absence of the information. If the robot has sufficient memory, it suffices to keep track of the entire structure of the field by doing a simple zig-zag exploration and build any structure. Due to the lack of memory, the robot can neither remember its past actions nor the positions of the bricks in the infinite grid. Since the grid does not have any markers, it is easy for a robot to get lost, and systematic exploration of an infinite grid is impossible with limited memory~\cite{DBLP:conf/focs/AleliunasKLLR79}. 
It is also important to keep track of the state of the structure and achieve termination after all bricks have been used by the robot.
\subsection{Our Contribution}
We consider the problem of the construction of a fort with a finite automaton in an infinite square grid. Specifically, we show that
\begin{itemize}
    \item any algorithm that builds a fort from a given initial connected field of span $s = O(\sqrt{z})$ and size, $z$ must use at least $\Omega(z^2)$ time.
    \item Algorithm~\ref{algo:buildfort} \textsc{BuildFort} builds a fort from a given initial connected field of span $s$ and size $z$ in $O(z^2)$ time. Hence, our algorithm is optimal.
\end{itemize}

\subsection{Related Works}
Many problems with robots or mobile agents exploring an unknown environment have been studied over the years~\cite{DBLP:series/lncs/Das19,DBLP:series/lncs/DaymudeHRS19}. 
The studies can broadly be categorized in two ways based on the environment in which the robot operates; it can be a  network designed as a graph along the edges of which the robot moves, or it can be a geographical terrain that might be filled with obstacles.
A problem of pattern formation has been studied in \cite{10.1145/1835698.1835761} and \cite{doi:10.1137/S009753979628292X}. It discusses robots that can freely move in the plane, arrange themselves in a pattern given as an input. In these studies, the robots are asynchronous and have full visibility of each others' positions. 
Programmable particles modeled as automatons are also considered in the literature~\cite{DBLP:journals/dc/LunaFSVY20,DBLP:conf/icdcn/LunaFPSV18,DBLP:conf/europar/LunaFSVY20}.
Luna et al.~\cite{DBLP:journals/dc/LunaFSVY20} considered shape formation by a set of programmable particles in a triangular grid.
With faulty particles, Luna et al.~\cite{DBLP:conf/icdcn/LunaFPSV18} showed the recovery of a line in a triangular grid.
The graph exploration problem can have two different scenarios depending on whether the nodes have a unique label or not. In the former case, the robot can identify different labels and thus uniquely identify each node. In our case, the grid is anonymous, so it is similar to an undirected anonymous graph. Such graph exploration problems have also been discussed~\cite{DBLP:conf/opodis/ChalopinDK10,DBLP:journals/algorithmica/CzyzowiczDP20}. 

\subsection{Organization}
The rest of the paper is organized as the following. In Section~\ref{sec:lowerbound}, we prove a lower bound to build a fort. In Section~\ref{sec:fortformation}, we present an algorithm that builds the fort with Section~\ref{sec:complexity} proving the correctness and determining the complexity of the algorithm. Finally Section~\ref{sec:conclusion} concludes the paper.

\section{Lower Bound to Build a Fort}\label{sec:lowerbound}

In this section, we prove the lower bound for the formation of a fort from an initially connected field to be $\Omega(z^2)$, where $z$ is the number of bricks. 
We show this lower bound by converting a configuration of bricks in a \textit{rough disc} to a fort.
A \textit{rough disc} is an arrangement of bricks such that the distance from any brick on the boundary from a center brick is either $r$ or $r+1$, where $r$ is the radius (measured with Manhattan distance) of the rough disc as shown in Fig.~\ref{fig:roughdisc}(a). A \textit{perfect disc} has all the bricks on the boundary at a distance $r$ from the center. 

\begin{figure}
    \centering
    \includegraphics[width=0.8\linewidth]{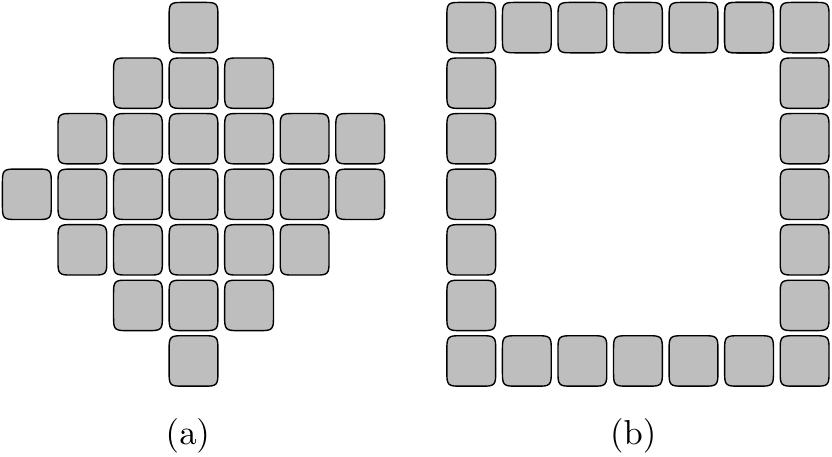}
    \caption{A rough disc and a fort}
    \label{fig:roughdisc}
\end{figure}

\begin{theorem}\label{theoremlowerbound}
There exists an initial field of size $z$ and span $s = O(\sqrt{z})$ such that any algorithm that builds a fort starting from this field must use time $\Omega(z^2)$.
\end{theorem}
\begin{proof}
Consider a rough disc $D_1$ with $z$ bricks and radius $r \geq 9$. Then $z$ satisfies $ z_1 \leq z < z_2$, where $z_1 = 2r^2 + 2r + 1$ and $z_2 = 2r^2 + 4r + 2~ (= 2r^2 + 2r + 1 + 2r + 1)$. The span of this initial arrangement of bricks is $s = 2r$ and $z = O(r^2)$, so we have $s = O(\sqrt{z})$.
Consider the two opposite walls of the resulting fort. 
The distance between the two opposite walls is $z/4$, and there are at least $z/4$ bricks on each wall.
Suppose the bricks in $D_1$ are at a distance $k$ from one wall, then they are at a distance at least $z/4 - k - s$ from the opposite wall.
Now, the bricks would occupy their target position on one of the walls.
So at least $z/4$ bricks would move a distance $k' = \max(k, z/4 - k -s)$. Since $s = O(\sqrt{z})$, we have $k' \geq z/10$ for $z \geq 100$.
Any algorithm must move at least $z/4$ bricks by a distance $z/10$ and must use $\Omega(z^2)$ time to build the fort.
\end{proof}
Note that the span of the resulting fort is $O(z)$. Also, the maximum span of the initial configuration can be at most $z$.
\section{Fort Formation}\label{sec:fortformation}

\begin{figure}[ht]
\centering
\includegraphics[width=0.6\linewidth]{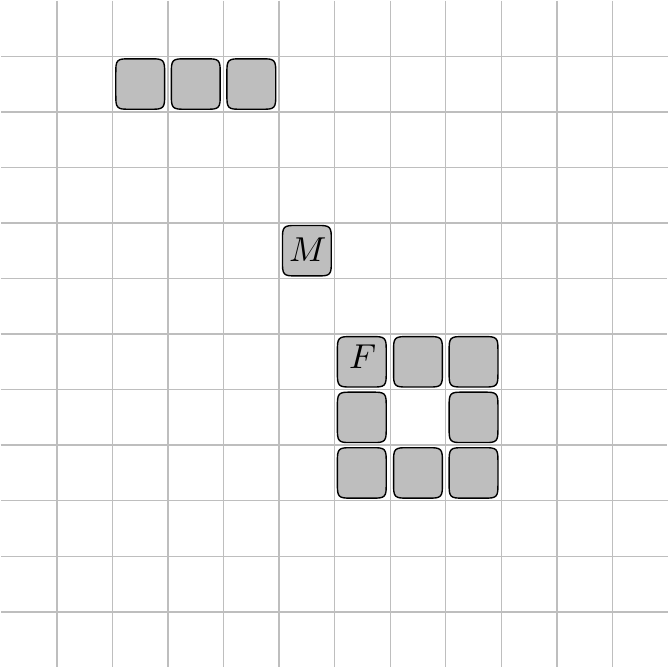}
\caption{A fort with 8 bricks, the marker and a free component. The cell labeled F is the first cell of the fort and the cell labeled M is the marker.}
\label{figfort8}
\end{figure}

The algorithm for fort formation works like the following.
The robot first creates two special components, the \textit{marker} and the \textit{first cell} of the fort, to initiate the construction of the fort (ref. Fig.~\ref{figfort8}). 
Any other component is called a \textit{free component}.
A marker is used to determine when the robot is in the vicinity of the fort. 
The marker is always positioned closer to the fort and farther from the free component.
The construction of a fort constitutes of four main steps:
\begin{enumerate}
    \item create enough area for the fort to extend
    \item pick a brick from the field
    \item bring the brick back to the fort
    \item add the brick to the fort to extend it
\end{enumerate}
To achieve this goal, we use three procedures developed by Czyzowicz et al.~\cite{DBLP:journals/algorithmica/CzyzowiczDP20} for nest formation as subroutine for fort formation, that are \sloppy \textsc{Sweep}, \textsc{FindNextBrick} and \textsc{ReturnToMarker}. 
We present a summary of those subroutines and refer the reader to the paper by Czyzowicz et al.~\cite{DBLP:journals/algorithmica/CzyzowiczDP20} for more details.

\begin{enumerate}[label=\Roman*., align=left, leftmargin=*]
    \item[\textsc{Sweep}:] This procedure is called every time a brick is added to the structure starting from the first brick. As the name suggests, in this procedure, the robot sweeps the bricks nearby the structure to a distance such that there is enough space to build the structure.
    Specifically, the robot makes a counterclockwise traversal of the fort and checks for full cells at a distance at most seven from the fort, which is not part of the fort. 
    It finds an empty cell at a distance at least seven from the fort, which is in a direction away from the fort and places the brick that it has picked from the full cell. This procedure may create multiple components. The components are created at a distance of seven while the robot has information about all the bricks up to distance eight; hence all the components remain accessible. 
    At the end of the procedure, the marker is placed at a distance three from the fort and a distance four from a component.
    \item[\textsc{FindNextBrick}:] Once the first cell of the fort is established, we need to add bricks to the fort to make it larger. This procedure finds the next brick that is to be added to the fort. It performs a \textit{switch-traversal} of a \textit{search-walk} in one of the components until it reaches a \textit{leaf cell}\footnote{A full cell with exactly one adjacent cell} and then picks the brick at the leaf cell. If it does not reach a leaf cell at the end of the search walk, it performs \textit{shifting} to get a free brick.
    \item[\textsc{ReturnToMarker}:] After the robot picks the free brick, it performs a reverse switch-traversal. It reaches near the marker at the end of the reverse switch-traversal, and then it returns to the marker.
\end{enumerate}

\begin{algorithm}[!ht]
\caption{\textsc{Sweep}}\label{proc:sweep}
\begin{algorithmic}[1]
\State $M$ = Marker
\If{the robot is at $M$}
    \State go to the first cell of the fort $F$
\EndIf
\State perform a full counterclockwise direction of the border of $F$, perform the following actions after every step:
\For{each full cell $c \notin F \cup \{M\} $ at distance at most 7 from the robot}
    \State $c'$ = the cell currently occupied by the robot 
    \State go to $c$ and pick the brick 
    \State move in direction away from $F$ and stop at the first empty cell at distance at least 7 from $F$
    \State drop and brick and return to $c'$
    \EndFor
\If{there exists a free component $C$}
    \State pick the brick from marker and place it at distance $3$ from the first cell of $F$ and at distance $4$ from $C$, creating a new marker
\EndIf
\State go to the new marker
\end{algorithmic}
\end{algorithm}

\begin{algorithm}[!ht]
\caption{\textsc{FindNextBrick}}\label{proc:FindNextBrick}
\begin{algorithmic}[1]
\State $W $ = the search walk that starts at the cell where the robot is located
\State go to the nearest cell belonging to a free component 
\State perform a switch-traversal of $W$
\If{the robot is at the a leaf}
    \State pick the brick
\Else
    \State perform shifting
\EndIf
\end{algorithmic}
\end{algorithm}

\begin{algorithm}[!ht]
\caption{\textsc{ReturnToMarker}}\label{proc:ReturnToMarker}
\begin{algorithmic}[1]
    \State $W$ = the search walk traversed in the last cell to \textsc{FindNextBrick}
    \State let $S_1, S_2,\cdots, S_l$ be the segments in $W$
    \If{the robot is at the first cell of $S_l$}
        \State turn towards the penultimate cell of $S_{l-1}$
        \State move $min\{2, |S_{l-1}| - 1\}$ cells forward
        \State if $|S_{l-1}| = 2$ and $l > 2$, then make a turn towards the penultimate cell of $S_{l-2}$
    \EndIf
    \State starting at the current location, perform a switch-traversal of $\psi(W)$
    \State go to the marker
\end{algorithmic}
\end{algorithm}

In Example~\ref{example:searchwalk}, we describe the switch traversal of a field and how the robot obtains a free brick to bring back to the marker. We consider a simple example. There can be much more complicated cases. 

\begin{figure}[ht]
    \centering
    \includegraphics[width=0.7\linewidth]{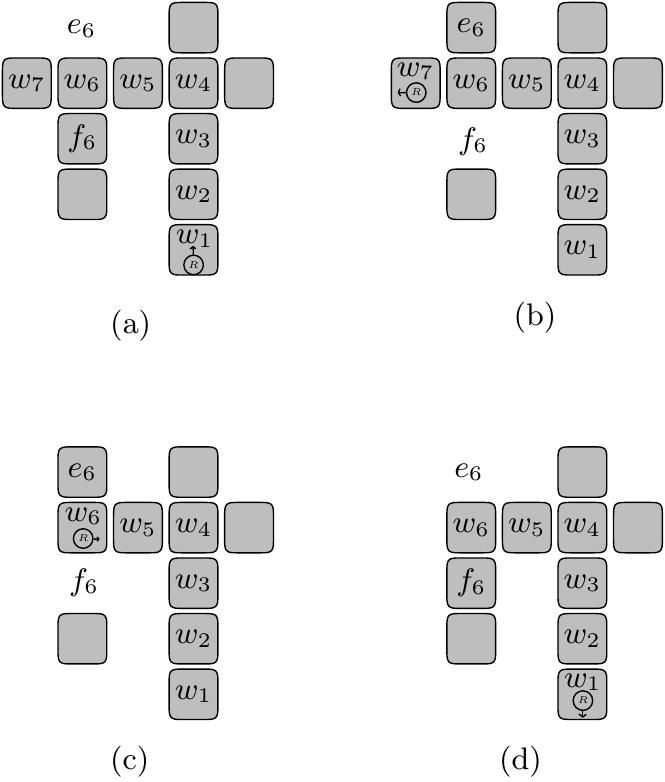}
    \caption{(a) The search-walk $\mathcal{W}$ starts at $w_1$, follows a left-free segment to $w_4$, and a right free segment up to the leaf node at $w_7$. (b) Switch traversal moves the brick at $f_6$ to $e_6$. (c) The robot picks the brick at leaf cell $w_7$. (d) The robot performs switch traversal on the reverse search walk of $\mathcal{W}$ to return to $w_1$.}
    \label{fig:searchwalk}
\end{figure}

\begin{example}\label{example:searchwalk}
    Consider a field, as shown in Fig.~\ref{fig:searchwalk}(a). The robot starts at $w_1$ and facing north. 
    The robot starts with a left-oriented traversal. It moves straight until it is possible to move left and then it turns left. After the turn, it switches its orientation to the right and moves straight until it is possible to move right.
    The search walk $\mathcal{W}$ starts at $w_1$ with a left-free segment $S_1 = (w_1, w_2, w_3, w_4)$ and then a right-free segment $ S_2 = (w_4, w_5, w_6, w_7)$ which ends at the leaf cell $w_7$. 
    Observe that, for the walk to proceed further, the robot has to turn right at $w_7$ or go straight, but there are no full cells after $w_7$.
    The robot performs a switch traversal, as shown in Fig.~\ref{fig:searchwalk}(b), and arrives at $w_7$.
    To keep track of the path it has traversed, it converts the left-free segments into right-free segments and vice versa. 
    As shown in Fig.~\ref{fig:searchwalk}(b), it moves the brick from left to right from $f_6$ to $e_6$. 
    Now, the reverse of the segment $S_2$ becomes a right-free segment starting from $w_7$.
    As per Fig.~\ref{fig:searchwalk}(c), it picks up the leaf cell at $w_7$ and turns around to reach $w_6$.
    Finally, while the robot performs the reverse search walk, it restores the brick from $e_6$ to $f_6$, and the connectivity of the component is reestablished.
    Thus we pick a brick from the component and arrive at the starting point $w_1$. 
\begin{figure}[ht]
    \centering
    \includegraphics[width=0.7\linewidth]{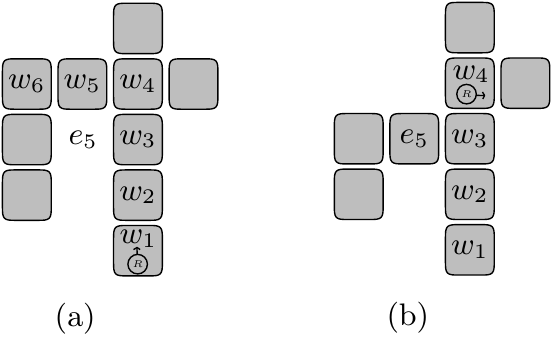}
    \caption{(a) The search-walk $\mathcal{W}$ starts at $w_1$, follows a left-free segment to $w_4$, and a right free segment up to the corner at $w_6$. (b) The robot picks the brick at leaf cell $w_6$ and performs shifting by moving the brick at $w_5$ to $e_5$.}
    \label{fig:shifting}
\end{figure}

    In the next call to \textsc{FindNextBrick}, the robot will again perform a switch traversal and arrive at $w_6$ as shown in Fig.~\ref{fig:shifting}(a). Now, it cannot turn right at $w_6$. So, the robot has to pick a brick from there, since the position of a leaf node may not be close to $w_6$. To get a free brick, the performs \textit{shifting} after picking the brick at $w_6$. 
    In shifting, the robot moves $w_5$ to $e_5$ and reaches $w_4$ carrying a brick as shown in Fig.~\ref{fig:shifting}(b). It makes sure that the component remains connected while a free brick is obtained by the robot, and then it performs a reverse switch traversal to return to $w_1$. It obtains a brick at the end of the traversal. \qed
\end{example} 

\subsection{Constructing the Fort}

As described in Example~\ref{example:searchwalk}, the Procedure \ref{proc:sweep} (\textsc{Sweep}), Procedure \ref{proc:FindNextBrick} (\textsc{FindNextBrick}) and Procedure \ref{proc:ReturnToMarker} (\textsc{ReturnToMarker}) work in tandem to provide a free brick. Next, we describe the construction of the fort as we add bricks one by one to extend the structure.
The robot carrying a free brick arrives at a marker after \textsc{ReturnToMarker}. From the marker, the robot comes to the first cell of the fort and invokes Procedure~\ref{proc:extendfort} (\textsc{ExtendFort}). Until the size of the fort is $4$, the robot follows a strict rule of adding bricks to the fort in the manner shown in Figure \ref{figstage0}.

\begin{figure}[ht]
\centering
\includegraphics[width=\linewidth]{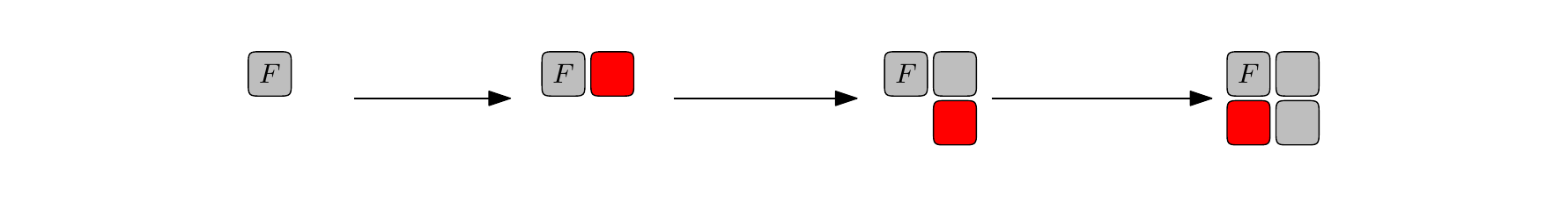}
\caption{The first four stages of construction of a fort where the newly added brick in each step is marked with a red shade.}
\label{figstage0}
\end{figure}

These configurations are handled when the value of the variable $stage = 0$. Once the fourth brick is added, the value of $stage$ becomes $1$.
Now the bricks are added in a manner that the resulting configuration is similar after every four steps. This can be seen in Figure \ref{figstage1}. 

\begin{figure}[ht]
\centering
\includegraphics[width=\linewidth]{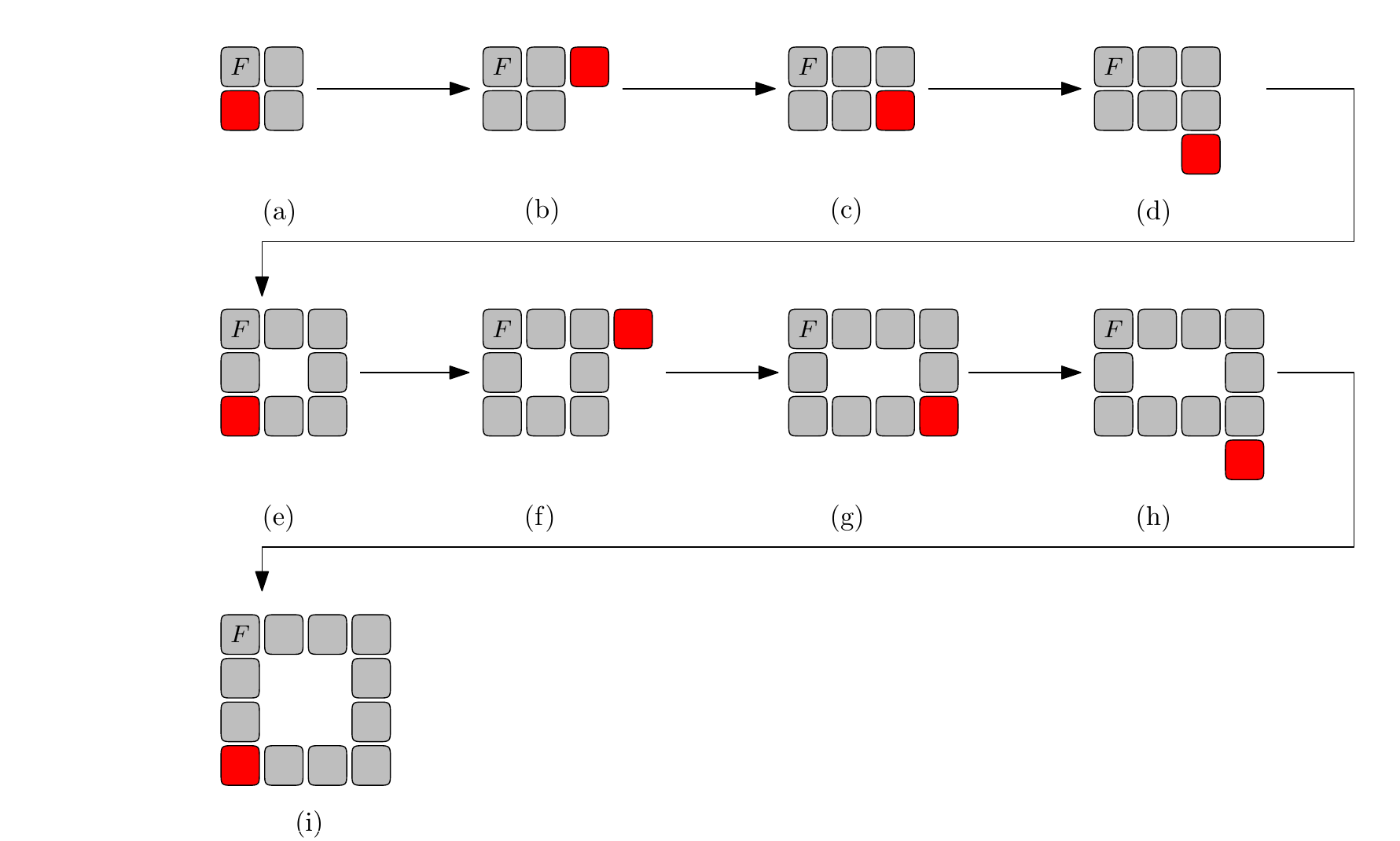}
\caption{The repeating configurations where (a), (e) and (i) are perfect forts; (c) and (g) are rectangular forts; (b), (d), (f) and (h) are rough forts.}
\label{figstage1}    
\end{figure}
\begin{algorithm}[h!t]
\caption{\textsc{ExtendFort}}\label{proc:extendfort}
\begin{algorithmic}[1]
\State Arrive at the first cell of the fort oriented to east
    \If{$stage=0$} 
        \If{$counter = 1$}
            \State Place the brick in the front
            \State Increase $counter$ by 1
        
        \ElsIf{$counter$ is 2}
            \State Move one step forward 
            \State Place brick on the right
            \State Increase $counter$ by 1 
        \ElsIf{$counter$ is 3}
            \State Place the brick on the right 
            \State Increase $counter$ by 1 
            \State Increase $stage$ by 1
        \EndIf
    \Else
        \If{$counter$ is 0}
            \State Call \textsc{TraverseWall}
            \State Drop Brick in Front
            \State Increase $counter$ by 1
        \ElsIf{$counter$ is 1}
            \State Call \textsc{TraverseWall}
            \State Move one step back
            \State Turn Right
            \State Call \textsc{ShiftBricks}
            \State Increase $counter$ by 1
        \ElsIf{$counter$ is 2}
            \State Call \textsc{TraverseWall}
            \State Turn Right
            \State Call \textsc{TraverseWall}
            \State Drop Brick in front
            \State Increase $counter$ by 1
        \ElsIf{$counter$ is 3}
            \State Call \textsc{TraverseWall}
            \State Turn Right
            \State Call \textsc{TraverseWall}
            \State Move one step back
            \State Turn Right
            \State Call \textsc{ShiftBricks}
            \State Increase $counter$ by 1
        \EndIf
    \EndIf
    \State $counter \leftarrow counter$ mod $4$
    \State Call \textsc{Sweep}
\end{algorithmic}
\end{algorithm}

In \textsc{ExtendFort}, one should note that turning right means the robot orients itself in the direction which is to its right originally. The same goes for turning around and turning left. 

A configuration of the fort is controlled by the variables $stage$ and $counter$. Depending on this configuration, the robot, starting from the first cell, traverses the perimeter of the fort and extends the fort. Sometimes, it needs to shift the bricks to the left. It does this by dropping the brick it currently carries to the left, picking the brick from the current cell, and moving forward. After adding the brick and changing the variables accordingly and \textsc{Sweep} is called.

\sloppy Procedure \textsc{ExtendFort} invokes Procedure~\ref{proc:TraverseWall}
(\textsc{TraverseWall}) and Procedure~\ref{proc:shiftbricks} (\textsc{ShiftBricks}). In \textsc{TrverseWall}, the robot moves along the wall of the fort until it reaches the end of the wall. In \textsc{ShiftBricks}, the robot shifts the bricks of the wall outwards. Since we traverse the walls of the fort in a clockwise manner, \textsc{ShiftBricks} always moves the bricks to the left of it.

\begin{algorithm}[ht]
\caption{\textsc{TraverseWall}}\label{proc:TraverseWall}
\begin{algorithmic}[1]
\While{the cell in front is full}
    \State move one step forward
\EndWhile
\end{algorithmic}
\end{algorithm}

\begin{algorithm}[ht]
\caption{\textsc{ShiftBricks}}\label{proc:shiftbricks}
\begin{algorithmic}[1]
\While{the cell in front is full}
    \State move one step forward
    \State drop the brick on the left
    \If{the cell in front is full}
        \State pick the brick from the current cell
    \EndIf
\EndWhile    
\end{algorithmic}
\end{algorithm}

\subsection{Algorithm \textsc{BuildFort}}

This section describes Algorithm~\ref{algo:buildfort} (\textsc{BuildFort}) for the robot to build a fort from a given initial connected field. First, the robot points out the marker and the first cell of the fort. Then \textsc{Sweep} is called to create space for fort formation, and the variables are declared. The robot brings bricks one by one and adds them to the fort to extend it. At last, the robot brings the brick to the marker and adds it to the fort. 
\makeatletter
\renewcommand*{\ALG@name}{Algorithm}
\makeatother
\begin{algorithm}[ht]
\caption{\textsc{BuildFort}}\label{algo:buildfort}
\begin{algorithmic}[1]
\If{the span of the field is at most 2}
    \State exit 
\EndIf
\State the cell occupied by the robot is at the marker.
\State a full cell at distance 2 from the marker becomes the first cell of the fort (i.e., the north-west corner cell of the fort)
\State Call \textsc{Sweep} 
\State Set $counter = 1$
\State Set $stage = 0$
\While{there exists a free component C}
    \State Call \textsc{FindNextBrick}
    \State Call \textsc{ReturnToMarker}
    \State Call \textsc{ExtendFort}
\EndWhile
\State pick the marker and add it using \textsc{ExtendFort}
\end{algorithmic}
\end{algorithm}

\section{Correctness and Complexity}\label{sec:complexity}

To prove the correctness of the procedures, we will define a few terms first.
A fort has a \textit{gap of width k} if each free component is at a distance $k+1$ from the fort. 
A field is structured if it satisfies the following conditions:
\begin{itemize}
    \item the marker is at a distance three from the first cell of the fort and a distance four from some free component.
    \item the fort has a gap of width seven.
\end{itemize}
If a component is at a distance larger than $7$ from the fort, it is called a \textit{lost component}.
A field is \textit{strongly structured} if it satisfies the following conditions:
\begin{itemize}
    \item the field is structured
    \item there are no lost components
    \item the robot is at the marker.
\end{itemize}

\subsection{Correctness}
We refer the reader to the paper by Czyzowicz \textit{et al.}~\cite{DBLP:journals/algorithmica/CzyzowiczDP20} for the correctness of the procedures \textsc{Sweep}, \textsc{FindNextBrick}, and \textsc{ReturnToMarker}.
After execution of \textsc{Sweep}, the robot moves the bricks nearby the fort to a distance at least 7 to create space for expansion. 
The robot takes the bricks from full cells within distance seven from the fort and drops them in an empty cell next to a full cell so that the brick remains part of a free component which is at most distance seven from the fort.
This may result in multiple components, but \textsc{Sweep} ensures that there are no lost components and properly repositions the marker such that a free component can always be found for the robot to execute \textsc{FindNextBrick} on the free component. 
We state the following lemma.

\begin{lemma}\label{lemsweepcorrect}
\textsc{Sweep} results in a strongly structured field. \cite{DBLP:journals/algorithmica/CzyzowiczDP20}
\end{lemma}

Next, we show the correctness of procedures \textsc{ShiftBricks}, \textsc{ExtendFort} before showing the correctness of Algorithm \textsc{BuildFort}. 

\begin{lemma}\label{lemextendfort2}
Procedure \textsc{ShiftBricks} 
correctly shifts the bricks on one wall of the fort outwards apart from the first and last brick of that wall.
\end{lemma}
\begin{proof}
The robot carrying a brick has three different situations when Procedure \textsc{ShiftBricks} is invoked. 
It shifts all the bricks to the left of its current orientation until it encounters an empty cell.
In Fig.~\ref{fig:shiftbricks}(a), the robot carries a brick. Since the cell in front of the robot is not an empty cell, it moves forward and drops the brick on its left, as shown in Fig.~\ref{fig:shiftbricks}(b). Now it picks up the brick from its current cell, and the resulting configuration of bricks looks like Fig.~\ref{fig:shiftbricks}(c). Finally, the procedure terminates when it encounters an empty cell as per Fig.~\ref{fig:shiftbricks}(d) and places the brick to its left and does not pick up the brick in its current cell.
Note that the configurations of bricks in Fig.~\ref{fig:shiftbricks}(a) and Fig.~\ref{fig:shiftbricks}(d) correspond to a rough fort with an odd number of bricks and a rectangular fort.

\begin{figure}
    \centering
    \includegraphics[width=\linewidth]{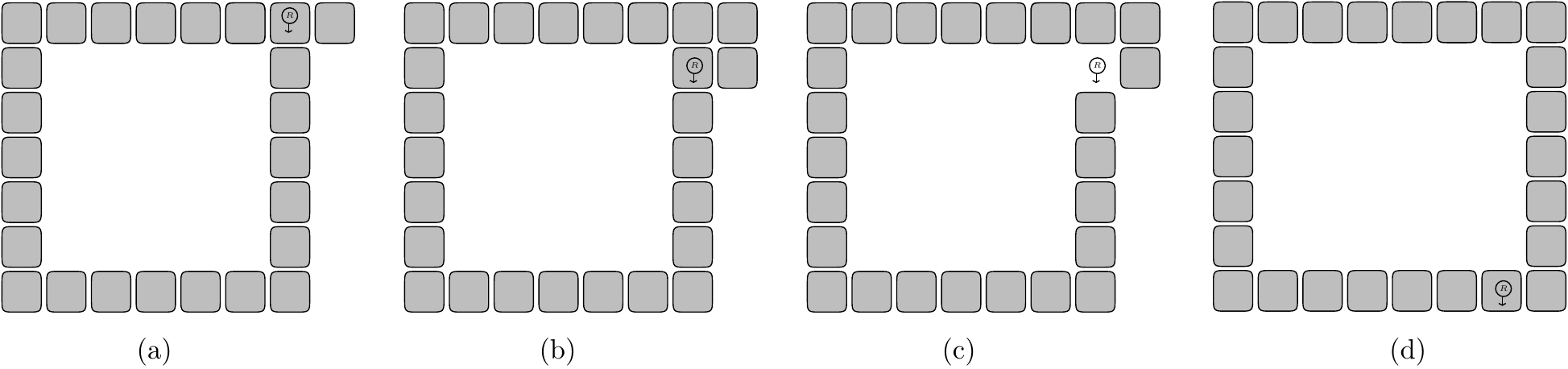}
    \caption{An execution of Procedure \textsc{ShiftBricks}}
    \label{fig:shiftbricks}
\end{figure}

If the cell in front of the current cell of the robot is empty, then the condition on line 1 of Procedure~\ref{proc:shiftbricks} returns false, and the procedure terminates.
If the first cell of the wall has an empty cell in front of it along the robot's orientation, it is also the last. If the first cell has a full cell in front, then line 2 gets executed, and the cell moves forward without picking the brick from the first cell of the row. Hence the first cell of the row is not shifted. This is due to the fact that the cell to the left of the first cell is full, and \textsc{ShiftBricks} is invoked to align the wall with the cell to its left.

The robot moves one cell forward according to line 2, and this cell contains a brick as it satisfied the condition in line 1. Now the brick is dropped to the left as per line 3 and will be the new wall. 
Since \textsc{ShiftBricks} is invoked after \textsc{Sweep}, the left cell would be empty since there is a gap of width seven from any nearest component.
Hence it can drop the brick it is carrying on its left in line 3 correctly. Now, since the robot does not carry any brick, it can pick the brick from the current cell if the condition in line 4 is satisfied, i.e., the robot is not at the last cell of the wall. Only the last cell of the wall will not have a full cell in the front, and thus all the cells of the wall will shift one cell to the left except the last one. 
\end{proof}

In the following lemma, we prove the correctness of \textsc{ExtendFort}.

\begin{lemma}\label{lemextendfortcorrect}
The robot correctly adds a brick to the fort in Procedure~\ref{proc:extendfort} (\textsc{ExtendFort}).
\end{lemma}
\begin{proof}
Before the start of Procedure~\ref{proc:extendfort} (\textsc{ExtendFort}), the robot is at the marker and carries a brick. This procedure works in two stages with $stage = 0$ and $stage =1$. 

When $stage$ is 0, the robot follows a defined set of instructions to place the first $3$ bricks to get a perfect fort of size $4$. 
In line 1, it arrives at the first cell of the fort since it is at a bounded distance from the marker, and no other brick is closer to the marker.
Initially, $counter$ is $1$, lines 4-5 get executed. Shifting done in \textsc{FindNextBrick} moves a brick by at most two cells, and from lemma \ref{lemsweepcorrect}, we know before \textsc{FindNextBrick}, the field was strongly structured. Thus the fort had a gap of width seven, and after its execution, the fort has a gap of width at least 5. 
Thus the cell on the east of the first cell should be empty, and hence the robot can place the brick in front as it is oriented to the east. The robot increases $counter$ by $1$ so that in the next call of \textsc{ExtendFort} lines 6-9 get executed. Since the fort now has two bricks, in line 8, the robot can move one cell forward and place the brick in the empty adjacent cell. 
The robot extends the fort and increases $counter$ by $1$. Thus, when the robot arrives to add the next brick, lines 10-14 get executed. The robot places the brick to get a perfect fort of span $2$ correctly since the gap is at least five, and thus, the cell on the south of the first cell of the fort must be empty. The robot increases $counter$ by $1$ and $stage$ by $1$. Thus, in every call to \textsc{ExtendFort} hereafter, the else block in lines 15-39 will get executed. Also, since line 40, updates $counter$ to be its remainder when divided by $4$, $counter$ becomes $0$ after every $4^{th}$ brick is added. 

When $stage$ is $1$, the robot adds the brick to the fort in a recursive manner, the shape of the fort being repeated after adding every $4^{th}$ brick. 
When the number of bricks in the fort is of the form $4t$ where $t$ is any positive integer and new brick has to be added, lines 16-19 get executed. The call to \textsc{TraverseWall} ensures that the robot reaches the end of the north wall of the fort. Since the gap is at least $5$, the robot places the brick in the empty cell extending the north wall of the fort by one brick. Now the number of bricks in the fort is of the form $4t + 1$, and $counter$ is $1$. When the next brick is added, lines 20-25 get executed. The call to \textsc{TraverseWall} ensures that the robot reaches the last cell of the north wall. Then it moves one step back and arrives at the penultimate cell, which is also the north of the east wall of the fort. It turns to the right and moves along the east wall. Lemma \ref{lemextendfort2} ensures that the call to \textsc{ShiftBricks} shifts all the bricks of this row to one cell to the left. As a result, the new brick is placed at the south-east corner of the fort. In the next call, the number of bricks is $4t + 2$, and $counter$ is $2$ results in the execution of lines 26-31. The first call to \textsc{TraverseWall} ensures that the robot reaches the last cell of the north wall, and then it turns right to the east wall where another call to \textsc{TraverseWall} ensures that it reaches the south of the east wall. The robot places the brick in the front, extending the east wall to the south by one brick. In the next call to \textsc{ExtendFort}, lines 32-39 get executed. After two calls to \textsc{TraverseWall} and a turn between them, the robot arrives at the southernmost cell of the east wall. Then it moves one step and turns right onto the south wall.
It invokes \textsc{ShiftBricks} procedure and moves the bricks on the south wall one step to the south (ref. Lemma \ref{lemextendfort2}). After this, the robot places brick on the south-west corner of the fort, extending the size of the fort to $4t + 4$ and $counter$ to $4$. After line 42 is executed, $counter$ resets to $0$, and the procedure runs repetitively. In line 43, Procedure \textsc{Sweep} ensures that the gap of the fort is of width $7$ before the robot goes on to pick another brick. 
Thus the procedure \textsc{ExtendFort} correctly adds a brick to the fort.
\end{proof}

\begin{theorem}
\textsc{BuildFort} correctly builds a fort.
\end{theorem}
\begin{proof}
To prove this theorem, we use induction on the number of iterations of the while loop in line 9 in $\textsc{BuildFort}$.
\textsc{Sweep} is called in line 6 before the while loop. From lemma \ref{lemsweepcorrect}, we know that the execution of \textsc{Sweep} has resulted in a strongly structured field. Thus before the first iteration, i.e., $i = 1$, of the while loop, the field is strongly structured. We assume the induction hypothesis holds after the $i^{th}$ iteration that the field is strongly structured and prove this remains true after the $(i+1)^{th}$ iteration. 
Using procedures \textsc{FindNextBrick} and \textsc{ReturnToMarker}, the robot carries a brick after it performs the search walk on one of the components. 
From Lemma~\ref{lemextendfortcorrect}, we know that the robot correctly adds a brick to the fort. Also, \textsc{Sweep} is called again in \textsc{ExtendFort}, which results in a strongly structured field. Thus, we proved our induction argument. Thus, the field is strongly structured before each iteration of the while loop in \textsc{BuildFort} and correctly builds the fort. 
\end{proof}

\subsection{Complexity}
From \cite{DBLP:journals/algorithmica/CzyzowiczDP20}, we have the complexity of \textsc{Sweep}, \textsc{FindNextBrick} and \textsc{ReturnToMarker} to be $O(s)$ for each of the procedures. We show the time complexity of Procedure~\ref{proc:extendfort} in Lemma~\ref{lemextendtime} and then we show the complexity of Algorithm~\ref{algo:buildfort} in Theorem~\ref{thm:buildforttime}.

\begin{lemma}\label{lemextendtime}
The execution of \textsc{ExtendFort} requires $O(z)$ time.
\end{lemma}
\begin{proof}
Since we add bricks to the fort starting from the north-west corner, adding a brick in the south-east corner requires the robot to travel a distance $s'$, where $s'$ is the span of the resulting fort. In case of the fort, $s' = O(z)$. Hence \textsc{ExtendFort} takes $O(z)$ time.
\end{proof}

\begin{theorem}\label{thm:buildforttime}
A fort is created by \textsc{BuildFort} given an initial connected field in time $O(z^2)$ where $z$ is the number of bricks in the field.
\end{theorem}
\begin{proof}
There are a total of $z$ bricks in the field. Initially, one brick was chosen as the fort and another as the marker. Thus the rest of $z-2$ bricks will be picked in the while loop of \textsc{BuildFort}. Thus there will be $z-2$ iterations of this loop. Also, we know that $O(s)$ is the length of the search walk, and the robot traverses it twice, once while picking the brick and the other while returning from it. We know that a single execution of \textsc{Sweep} takes time $O(s)$. Also, we saw in lemma \ref{lemextendtime} that adding a brick to the fort requires time $O(z)$, which is done in one execution of \textsc{ExtendFort}. 
The span of the initial configuration can be at most $z$.
Thus, \textsc{BuildFort} requires $O(z^2)$ time.
\end{proof}

\section{Conclusion and Future Work}\label{sec:conclusion}

We proposed a Fort formation problem by a robot modeled as a finite automaton in an infinite grid and established a lower bound for the same.
We developed an algorithm that builds a Fort from a given initial connected field of bricks in an infinite grid by a mobile robot in worst-case optimal time.
The research area is quite nascent, and a lot can be done further. 
One can explore in the direction of building different structures given an initial field. 
A variation of the bricks with hexagonal tiling can also be considered instead of a square bricks.
Extension to 3 dimensions with six directions of movement can also be considered as a future work.
Further research using multiple mobile agents to build the structures can be performed to determine if the time complexity can be improved.

\bibliographystyle{plain}
\bibliography{bib}

\begin{thebibliography}{10}

\bibitem{DBLP:conf/focs/AleliunasKLLR79}
Romas Aleliunas, Richard~M. Karp, Richard~J. Lipton, L{\'{a}}szl{\'{o}}
  Lov{\'{a}}sz, and Charles Rackoff.
\newblock Random walks, universal traversal sequences, and the complexity of
  maze problems.
\newblock In {\em 20th Annual Symposium on Foundations of Computer Science, San
  Juan, Puerto Rico, 29-31 October 1979}, pages 218--223. {IEEE} Computer
  Society, 1979.

\bibitem{DBLP:conf/opodis/ChalopinDK10}
J{\'{e}}r{\'{e}}mie Chalopin, Shantanu Das, and Adrian Kosowski.
\newblock Constructing a map of an anonymous graph: Applications of universal
  sequences.
\newblock In Chenyang Lu, Toshimitsu Masuzawa, and Mohamed Mosbah, editors,
  {\em Principles of Distributed Systems - 14th International Conference,
  {OPODIS} 2010, Tozeur, Tunisia, December 14-17, 2010. Proceedings}, volume
  6490 of {\em Lecture Notes in Computer Science}, pages 119--134. Springer,
  2010.

\bibitem{DBLP:journals/algorithmica/CzyzowiczDP20}
Jurek Czyzowicz, Dariusz Dereniowski, and Andrzej Pelc.
\newblock Building a nest by an automaton.
\newblock {\em Algorithmica}, 2020.

\bibitem{DBLP:series/lncs/Das19}
Shantanu Das.
\newblock Graph explorations with mobile agents.
\newblock In Paola Flocchini, Giuseppe Prencipe, and Nicola Santoro, editors,
  {\em Distributed Computing by Mobile Entities, Current Research in Moving and
  Computing}, volume 11340 of {\em Lecture Notes in Computer Science}, pages
  403--422. Springer, 2019.

\bibitem{10.1145/1835698.1835761}
Shantanu Das, Paola Flocchini, Nicola Santoro, and Masafumi Yamashita.
\newblock On the computational power of oblivious robots: Forming a series of
  geometric patterns.
\newblock In {\em Proceedings of the 29th ACM SIGACT-SIGOPS Symposium on
  Principles of Distributed Computing}, PODC ’10, page 267–276, New York,
  NY, USA, 2010. Association for Computing Machinery.

\bibitem{DBLP:series/lncs/DaymudeHRS19}
Joshua~J. Daymude, Kristian Hinnenthal, Andr{\'{e}}a~W. Richa, and Christian
  Scheideler.
\newblock Computing by programmable particles.
\newblock In Paola Flocchini, Giuseppe Prencipe, and Nicola Santoro, editors,
  {\em Distributed Computing by Mobile Entities, Current Research in Moving and
  Computing}, volume 11340 of {\em Lecture Notes in Computer Science}, pages
  615--681. Springer, 2019.

\bibitem{DBLP:conf/icdcn/LunaFPSV18}
Giuseppe Antonio~Di Luna, Paola Flocchini, Giuseppe Prencipe, Nicola Santoro,
  and Giovanni Viglietta.
\newblock Line recovery by programmable particles.
\newblock In Paolo Bellavista and Vijay~K. Garg, editors, {\em Proceedings of
  the 19th International Conference on Distributed Computing and Networking,
  {ICDCN} 2018, Varanasi, India, January 4-7, 2018}, pages 4:1--4:10. {ACM},
  2018.

\bibitem{DBLP:conf/europar/LunaFSVY20}
Giuseppe Antonio~Di Luna, Paola Flocchini, Nicola Santoro, Giovanni Viglietta,
  and Yukiko Yamauchi.
\newblock Mobile {RAM} and shape formation by programmable particles.
\newblock In Maciej Malawski and Krzysztof Rzadca, editors, {\em Euro-Par 2020:
  Parallel Processing - 26th International Conference on Parallel and
  Distributed Computing, Warsaw, Poland, August 24-28, 2020, Proceedings},
  volume 12247 of {\em Lecture Notes in Computer Science}, pages 343--358.
  Springer, 2020.

\bibitem{DBLP:journals/dc/LunaFSVY20}
Giuseppe Antonio~Di Luna, Paola Flocchini, Nicola Santoro, Giovanni Viglietta,
  and Yukiko Yamauchi.
\newblock Shape formation by programmable particles.
\newblock {\em Distributed Comput.}, 33(1):69--101, 2020.

\bibitem{doi:10.1137/S009753979628292X}
Ichiro Suzuki and Masafumi Yamashita.
\newblock Distributed anonymous mobile robots: Formation of geometric patterns.
\newblock {\em SIAM Journal on Computing}, 28(4):1347--1363, 1999.

\end{thebibliography}

\end{document}